\documentclass[a4paper,12pt]{amsart}
\usepackage{amsfonts}
\usepackage{rotating}
\usepackage{chapterbib}
\usepackage{cases}
\usepackage{graphicx, color}
\usepackage[latin1]{inputenc}
\usepackage[english]{babel}
\usepackage{amssymb}
\usepackage{amsmath, eurosym, layout, hyperref} 
\usepackage{times}
\usepackage{amsthm}
\usepackage{linegoal}
\newcommand{\IR}{\mathbb{R}}
\newcommand{\IZ}{\mathbb{Z}}
\newcommand{\IN}{\mathbb{N}}

\newcommand{\IC}{\mathbb{C}}
\newcommand{\IE}{\mathbb{E}}
\newcommand{\IP}{\mathbb{P}}
\newcommand{\IV}{\mathbb{V}}

\newcommand{\sA}{\mathcal{A}}

\newcommand{\supp}{\text{\rm{supp}}}

\newtheorem{theorem}{Theorem}

\newtheorem{lemma}[theorem]{Lemma}

\theoremstyle{definition}

\newtheorem{remarks}[theorem]{Remarks}

\newtheorem{definition}[theorem]{Definition}

\newtheorem*{assumption}{Assumptions}

\newcommand{\bew}{\noindent\textbf{Proof:}\quad}
\newcommand{\ebew}{\hfill\qed\\}

\renewenvironment{proof}{\bew}{\ebew}

\newenvironment{balign*}[1][12pt]{\setlength{\jot}{#1}\nonumber\align}{\endalign}

\renewcommand{\epsilon}{\varepsilon}

\begin{document}
\title{Lifshits tails for squared potentials}
\author[W.~Kirsch]{Werner Kirsch}
\author[G.~Raikov]{Georgi Raikov}

\begin{abstract}\setlength{\parindent}{0mm}
We consider Schr\"{o}dinger operators with a random potential which is the square of an alloy-type potential. We investigate their integrated density of states and prove Lifshits tails.

Our interest in this type of models is triggered by an investigation of randomly twisted waveguides.
\end{abstract}

\maketitle

{\bf  AMS 2010 Mathematics Subject Classification:} 82B44,  35R60, 47B80,  81Q10\\

{\bf  Keywords:} Integrated density of states, Lifshits tails, Squares of random potentials

\section{Introduction}
In the 1960s Lifshits \cite{Lifshits1} discovered that the density of states for periodic systems and the one for random systems show very different behavior near the bottom of their spectra. While the integrated density of states  $N (E)$ for a $d$-dimnesional periodic system behaves  like $(E - E_0)^{\frac{d}{2}}$, $E \searrow E_0$, near the ground state energy $E_0$, $N(E)$ it behaves like $e^{- C\,(E - E_0)^{- \frac{d}{2}}}$ for typical random systems. In the former case  the  integrated density of states is said to have {\em a van Hove singularity} at $E_0$, and in the latter one $N(E)$ exhibits {\rm a Lifshits tail} near $E_0$.

Starting with the seminal work by Donsker and Varadhan \cite{DonskerV}, there has been a strong interest in this type of questions in the mathematical physics literature. For a review (as of 2006) see e.g. \cite{KirschM} (see also \cite{AizenmanW}, \cite{Kirsch}), some more recent developments are \cite{Ghribi}, \cite{KloppN}, \cite{KloppN1} and \cite{Shen}.

One of the most common random potentials and the one we are dealing with in this paper is the alloy-type potential
\begin{equation}
U_\omega (x) = \sum_{i \in \IZ^d} q_i (\omega) f (x - i)
\label{eq:alloy}
\end{equation}
where $x \in \IR^d$, $q_i$ are independent, identically distributed random variables and $f$ is a (say) bounded measurable function decaying sufficiently fast at infinity.

Lifshits tails for
\begin{equation}
H_\omega = - \Delta + U_\omega
\end{equation}
are well known for alloy-type potentials as in \eqref{eq:alloy} if both the $q_i$ and the function $f$ have definite sign.

Recently, there has been interest in the case that $q_i$ and/or $f$ change sign (see e.g. \cite{Ghribi}, \cite{KloppN}, \cite{KloppN1}). In these models the lack of monotonicity makes it much harder to prove Lifshits tails.

In the paper \cite{KirschKR} David Krej\v{c}i\v{r}\'ik and the present authors investigate twisted wave guides $\mathcal{M}$ which emerge from the cylinder $M = m \times \IR$ with a cross section $m\subset\IR^{2}$ by rotation of $m$ around the axis cylinder at an angle $\theta$ which depends on the variable along this axis. 
The twist function $U(x):=\dot{\vartheta}(x)$, $x \in \IR$,  is supposed to be random. If the cross section is not rotationally symmetric and
its diameter is small, we were able to bound the integrated density of states $N$ of the Laplacian on the twisted waveguide by the integrated density of states of a one-dimensional
Schr\"{o}dinger operator with potential $V_{\omega}(x)= U_\omega (x)^2$.

In fact, Lifshits tails for the twisted waveguide correspond to Lifshits tails of the Schr\"odinger operator
\begin{equation}
H_\omega = - \Delta + V_\omega = - \Delta + {U_\omega}^{2}
\label{eq:square}
\end{equation}
with an alloy-type potential $U_\omega$ as in \eqref{eq:alloy}.
This observation was the initial motivation for the present paper.
In \cite{KirschKR} we need only the one-dimensional case of \eqref{eq:square} but here we will deal with this model in arbitrary dimension $d \geq 1$ as this will not cause additional complications.

Obviously, the potential $V_\omega (x) = U_\omega (x)^2$ is non-negative. We will allow, however, that both $q_i$ and $f$ may change sign, so that we lose monotonicity in those parameters.

\section{Setting}
We consider the random potential
\begin{equation}\label{eq:defU}
U_\omega (x) = \sum_{i \in \IZ^d} q_i (\omega) f(x - i), \quad  x \in \IR^d, \quad d\geq 1,
\end{equation}
on a probability space $(\Omega,\sA,\IP)$. The expectation with respect to $\IP$ will be denoted by $\IE$.
 Throughout this paper
we make the following assumptions.
\begin{assumption}\mbox{}\label{Ass}
\begin{enumerate}
\item The real valued random variables $q_i$ are independent and identically distributed. Their common distribution is denoted by $P_0$.\label{Ass:0}
\item The support $\supp\; P_{0}$ contains
more than one point, $0\in\supp\;P_{0}$ and $\supp ~ P_0 \subset [-Q,Q]$ for some $Q<\infty$.\label{Ass:1}
\item For some $K\geq 0, C>0$ and all $\epsilon>0$ small enough \label{Ass:2}
\begin{align*}
   \IP ( |q_i| < \epsilon ) \geq C \epsilon^K\,.
\end{align*}
\item $f$ is a bounded (measurable) real valued function, $f \neq 0$, with \label{Ass:3}
\begin{align*}
   |f(x)|~\leq~\frac{C}{(1+|x|)^{\alpha}}
\end{align*}
for some $C$ and $\alpha > d$.
\end{enumerate}
\end{assumption}

We remark that Assumption \ref{Ass:3} ensures that $\IE(|\sum_{i\in\IZ_{d}}q_{i}f(x-i)|)<\infty$, thus $V_{\omega}(x)$ exists
and is finite almost surely and for almost all $x$ (see e.\,g. \cite{KirschMa1}).

Next, we set
\begin{align}\label{eq:defV}
     V_\omega (x) = U_\omega (x)^2 = \Big(\sum q_i f (x-i)\Big)^2\,.
\end{align}
and define the operator
\begin{equation}\label{eq:defH}
H_\omega : = H_0 + V_\omega\, 
\end{equation}
with $H_0 : = - \Delta  $. Since $V_{\omega}$ is non-negative, it is clear that $\sigma (H_\omega) \subset [0, \infty)$. From general results we even have
$\sigma (H_\omega) = [0, \infty)$ (see \cite{KirschMa}). We remark that we made no assumption on the sign of the $q_{i}$ or of $f$. In fact, unless otherwise stated,
both may change sign.

Let us introduce the integrated density of states $N$ of $H_{\omega}$.
For $\Lambda = [-\frac{L}{2}, \frac{L}{2}]^d$ let $ H_{\Lambda}^N$ and $H_{\Lambda}^D$ be the operator $H_\omega$ restricted to $L^2 (\Lambda)$ with Neumann, resp. Dirichlet, boundary conditions. These operators have a purely discrete spectrum.
By $\lambda_k (H_{\Lambda}^N)$ and $\lambda_k (H_{\Lambda}^D)$ we denote the eigenvalues of $H_{\Lambda}^N$ respectively $H_{\Lambda}^D$ in increasing order and counted according to multiplicity.

For $E \in \IR$, we define
\begin{align*}
   N(H_{\Lambda}^N,E)~:=&~\# \{\lambda_k (H_{\Lambda}^N) \leq E \}\,,\\
   N(H_{\Lambda}^D,E)~:=&~\# \{\lambda_k (H_{\Lambda}^D) \leq E \}\,.
\end{align*}

Then  \emph{the integrated density of states} of $H_\omega$ is the limit
\begin{equation*}
N(E)  = \underset{L \rightarrow \infty}{\lim} \frac{1}{L^d} N(H_{\Lambda}^N,E)
 = \underset{L \rightarrow \infty}{\lim} \frac{1}{L^d} N(H_{\Lambda}^D,E)\,.
\end{equation*}

By \emph{Lifshits tails} we mean that the integrated density of states $N$ of $H_{\omega}$ behaves \emph{roughly} like $e^{-C\,(E-E_{0})^{-\gamma}}$ as $E\searrow E_{0}$
where $E_{0}$ is the
bottom of the spectrum of $H_{\omega}$. More precisely, we have 
\begin{align}\label{eq:Lifshits}
   \lim_{E\searrow E_{0}}\;\frac{\ln|\ln N(E)|}{\ln E}~=~-\gamma
\end{align}
where $\gamma>0$ is called the \emph{Lifshits exponent}. For Schr\"odinger operators $-\Delta + U_\omega$ with \emph{alloy-type potentials} $U_\omega$ as in \eqref{eq:defU} the Lifshits exponent depends on the behavior of $f$ at infinity.
If $|f(x)|\leq C|x|^{-(d+2)}$ for large $|x|$, then $\gamma=\frac{d}{2}$, the `classical' value for $\gamma$. If $f(x)\sim C|x|^{-\alpha}$ for $d<\alpha<d+2$ then
$\gamma=\frac{d}{\alpha-d}$ (see e.g. \cite{KirschS}).

\section{Results}
In this section we state our results for the squared random potential $V_\omega = U_\omega^2$ as in \eqref{eq:defV}. As in the conventional case (i.\,e. for \eqref{eq:defU}) we obtain
Lifshits behavior as in \eqref{eq:Lifshits}. Again, the Lifshits exponent depends on the behavior of $f$ at infinity. This time, however, the threshold is
$\alpha=d+1$ rather than the `conventional' $d+2$.

\begin{theorem}\label{th:lifshits}
\mbox{}
Suppose $q_{i}$ are independent random variables with common distribution $P_{0}$ satisfying Assumptions \ref{Ass:1} and \ref{Ass:2}.

{\rm (i)} If $f$ satisfies Assumption \ref{Ass:3} with some $\alpha\geq d+1$ then
\begin{equation} \label{g8}
\underset{E\searrow 0}{{\lim}} ~~ \frac{\ln | \ln N (E)|}{\ln E} ~=~- \frac{d}{2}\,.
\end{equation} 

{\rm (ii)} If ${\rm supp}\,P_0 \subset [0,\infty)$ and $f$ satisfies
\begin{align*}
  \frac{C_1}{(1+|x|)^{\alpha}}~\leq~ f(x)~\leq~\frac{C_{2}}{(1+|x|)^{\alpha}}
\end{align*}
for some $d < \alpha < d+1$ and constants $C_{1},C_{2}>0$, then
\begin{equation} \label{g9}
\underset{E \searrow 0}{{\lim}} ~~ \frac{\ln | \ln N (E)|}{\ln E} = - \frac{d}{2 (\alpha - d)}\,.
\end{equation}
\end{theorem}

\begin{remarks}
\mbox{}
\begin{enumerate}
\item For 'non-squared' random potentials as in \eqref{eq:defU} the critical value of $\alpha$ is $d + 2$, rather than $d + 1$ for the squared case.
\item We will prove Theorem \ref{th:lifshits} by showing corresponding upper and lower bounds on $N(E)$. Assumptions \ref{Ass:1} and \ref{Ass:2} are only needed
for the lower bounds. 
\end{enumerate}
\end{remarks}

\section{Strategy of the Proof}
We use the technique of Dirichlet-Neumann-bracketing (see \cite{KirschMa2} and \cite{KirschS}). This method is based on the inequalities
\begin{align}\label{eq:DNbrack}
   \frac{1}{|\Lambda|}\,\IE\big(N(H_{\Lambda}^{D},E)\big)~\leq N(E)~\leq ~ \frac{1}{|\Lambda|}\,\IE\big(N(H_{\Lambda}^{N},E)\big)
\end{align}
which are valid for any cube $\Lambda=\Lambda_{L}:=[-\frac{L}{2}, \frac{L}{2}]^{d}$ with $|\Lambda_{L}|=L^d$ being the volume of $\Lambda_{L}$.
Most of the time we simply write $\Lambda$ instead of $\Lambda_{L}$ as we did in \eqref{eq:DNbrack}.

The right hand side of \eqref{eq:DNbrack} can further be estimated by
\begin{equation} \label{g2}
   \frac{1}{|\Lambda|}\,\IE\big(N(H_{\Lambda}^{N},E)\big)~\leq~\frac{N(-\Delta_{\Lambda}^{N},E)}{|\Lambda|}\ \ \IP\big(\lambda_{1}(H_{\Lambda}^{N})\;<\;E\big)\,.
\end{equation} 
If $E \in (0, cL^{-2})$ with a constant $c>0$, then, obviously, 
\begin{equation} \label{g16}
N(-\Delta_{\Lambda}^{N},E) \leq N(-\Delta_{\Lambda_1}^{N},c) < \infty\,.
\end{equation}  
Consequently, we have to estimate
\begin{align}\label{eq:prob}
\IP\big(\lambda_{1}(H_{\Lambda}^{N})\;<\;E\big)
\end{align}
 from above.

To do so, we use the McDiarmid inequality which we introduce in Section \ref{sec:McDiarmid}.
The estimate of \eqref{eq:prob} using the McDiarmid inequality is done in Section \ref{sec:UB}.

In Section \ref{sec:LB} we estimate the left hand side of \eqref{eq:DNbrack} for a lower bound of $N$.

\section{Upper Bound}\label{sec:UB}
\subsection{Analytic estimate}\mbox{}

For the upper bound we use a perturbative approach following an idea of Stollmann \cite{Stollmann1}.

We set $H_{\Lambda}^{N} (t) := - \Delta_{\Lambda}^{N} + t\, V_\omega$ on $\Lambda$ with Neumann boundary conditions.
In the following we always take $\Lambda:=\Lambda_{L}$ the cube of side length $L$ around the origin. $L$ will be determined later.

By
\begin{align}
   E (t) = \lambda_1 (H_{\Lambda}^{N} (t))
\end{align}
we denote its lowest eigenvalue, and by $\varphi_0$ `the' normalized ground state of $H_{\Lambda}^{N} (0)$.
Note that $E$ is monotone increasing for $0 \leq t \leq 1$, and
\begin{equation*}
E (0) = \lambda_1 (- \Delta_{\Lambda}^{N}) = 0, \quad E (1) = \lambda_1 (H_{\Lambda}^{N}) \,.
\end{equation*}
Moreover, $E(\zeta)$ is a holomorphic function in\, $\{\zeta \in \IC ~ | ~~ |\zeta| < \nu \}$ for $\nu$ small, namely for $\nu = \pi L^{-2} =  \lambda_2 (-\Delta_{\Lambda}^{N})$. We have
\begin{equation*}
E' (0) = ~~ < \varphi_0, V_\omega \varphi_{0} > ~~ = \frac{1}{|\Lambda|} \int_{\Lambda} V_\omega (x) ~dx
\end{equation*}
by the Hellmann-Feynman Theorem (see e.\,g. \cite{ReedS}, Theorem XII.8, and the calculation after Theorem XII.3, or consult \cite{Stollmann2}, Theorem 4.1.29 ).

Consequently, the expectation of the random variable $E'(0)$ is given by
\begin{align}\label{eq:Epos}
   \IE(E'(0))~=~\frac{1}{|\Lambda|} \IE\Big(\int_{\Lambda} V_\omega (x)~dx\Big)~=~\IE\Big(\int_{\Lambda_{1}} V_\omega (x)~dx\Big)
\end{align}
where as usual $\Lambda=[-L/2,L/2]^{d}$ and $\Lambda_{1}$ is the unit cell $[-1/2,1/2]^{d}$.
It follows that $\IE(E'(0)) $ is strictly positive and independent of $L$.
By the analytic perturbation theory we also have:

\begin{lemma} {\rm (\cite{Stollmann2}, Lemma 2.1.2)} There are constants $C_{1},C_{2}$ such that for $0 \leq t \leq C_1 L^{-2}$ we have
\begin{equation*}
|E (t) - t E^{'} (0) | \leq C_2 L^2 t^2\,.
\end{equation*}
\end{lemma}
So, for $t \leq C_1 L^{-2}$, and $b$ to be chosen later, we obtain
\begin{equation}\label{eq:et}
\IP ( E (t) \leq b L^{-2}) \leq \IP (E^{'} (0) \leq C_2 L^2 t + b\, t^{-1} L^{-2})\,.
\end{equation}
The choice $t = \frac{b^{1/2}}{{C_{2}}^{1/2}} L^{-2}$, with $b$ small enough to guarantee $t\leq C_1 L^{-2}$,  makes the right hand side of \eqref{eq:et} smaller than
\begin{align}\label{eq:PEp} \IP \big( E^{'} (0)~ \leq~ 2\,{C_{2}}^{1/2}\,b^{1/2}\big)\,.
\end{align}

By \eqref{eq:Epos} and by decreasing $b$ further, if necessary, we can finally bound \eqref     {eq:PEp} by
\begin{align*}
   \IP \Big(| E^{'} (0) - \IE (E^{'} (0)) | > \frac{1}{2}\,\IE \big(E^{'} (0)\big)\Big)\,.
\end{align*}

Summarizing, we obtain for small $b$:
\begin{align}\label{eq:lade}
   \IP\big(\lambda_{1}(H_{\Lambda}^{N})\;<\;b\,L^{-2}\big)~\leq~    \IP \Big(| E^{'} (0) - \IE (E^{'} (0)) | > \frac{1}{2}\,\IE \big(E^{'} (0)\big)\Big)\,.
\end{align}

To estimate $\IP (\lambda_1 (H_{\Lambda}^N) \leq E )$, we may therefore estimate large deviations of the random variable
$$E^{'} (0)~=~\int_{\Lambda} V_\omega (x) ~dx $$
from its mean value
as long as we take
\begin{align}\label{eq:EL2}
   E~\sim~L^{-2}\,. 
\end{align}
In the following we choose $E$ respectively $L$ so that \eqref{eq:EL2} is satisfied.

To estimate large deviations of $\int_{\Lambda} V_\omega (x) ~dx $, 
we employ the McDiarmid inequality which we introduce in the following section.

\subsection{McDiarmid inequality}\label{sec:McDiarmid}\mbox{}

To estimate $\IP (\frac{1}{|\Lambda|} \int_{\Lambda} V_\omega (x) ~dx < \lambda)$ from above we will use
a concentration inequality due to McDiarmid in a slightly extended form.

\begin{definition}
   Let $I$ be a countable index set and for each $i\in I$ let $R_{i}$ be a subset of $\IR$.

A measurable function $F:\prod_{i\in I}R_{i}\longrightarrow \IR$ is called a \emph{McDiarmid function} if there are constants
   $\sigma_{j} \geq 0$ with $\sum_{j\in I}\sigma_{j}<\infty $ such that for all $X=\{x_{i}\}_{i\in I}\in\prod_{i\in I}R_{i}$
   and  $X'=\{x^{'}_{i}\}_{i\in I}\in\prod_{i\in I}R_{i}$ with $x_{i}=x^{'}_{i}$ for $i\not=j$, we have 
   \begin{equation}\label{eq:diffF}
|F (X) - F (X^{'}) | \leq \sigma_j \,.
\end{equation}
\end{definition}

\begin{theorem} 
Suppose $\{X_n \}_{n \in \IN}$ is a sequence of independent real valued random variables such that $X_n$ takes values in $R_n \subset \IR$.

Let $F: \prod_{n \in \IN} R_n \rightarrow \IR$ be a McDiarmid function with constant $\sigma_{n} $ and
set 
$$
\sigma^2 := \sum_{j \in \IN} \sigma_j^2.
$$
Then for all $\lambda > 0$, we have 
\begin{equation} \label{g15}
\IP (|F (X) - \IE ( F (X)) | > \lambda) \leq 2 e^{- 2 \frac{\lambda^2}{\sigma^2}}\,.
\end{equation}
\end{theorem}
\begin{proof}
This theorem, original from \cite{McDiarmid}, can be found in various sources, for example in \cite{Tao}, but only for finite collections $\{X_i \}_{i = 1}^{M}$ of random variables.
The 'limit $M \rightarrow \infty$' can be taken in the following way. 
Consider the vector $\underline{X}_M = (X_1, \cdots, X_M)$ of random variables and the non random vector
\begin{equation*} 
\underline{Y}^M := (Y_{M+1}, Y_{M+2}, \cdots) \in \prod_{n = M+1}^{\infty} R_n\,.
\end{equation*}
Set $F_M (\underline{X}) := F (X_1, \cdots , X_M, Y_{M+1}, \cdots )$ and $\IE_M = \IE (F_M (\underline{X}))$.
Note that both $F_M$ and $\IE_M$ depend on $Y^M$.
Using \eqref{eq:diffF}, we get 
\begin{equation*}  
|F (\underline{X}) - F_M (\underline{X})| \leq \sum_{n = M+1}^{\infty} \sigma_n \rightarrow 0
\end{equation*}
and
\begin{equation*}
|\IE \big( F (\underline{X})\big) - \IE_M | \leq \sum_{n = M+1}^{\infty} \sigma_n \rightarrow 0
\end{equation*}
as $M \to \infty$, uniformly in $\underline{Y}^M$. Now,
\begin{align*}
& \IP\,\Big(|F (\underline{X}) - \IE (F (\underline{X})) | > \lambda\Big)\notag\\
 \leq &~ \IP\, \Big(|F(\underline{X}) - F_N (\underline{X} ) | + |F_M (\underline{X}) - \IE_M |
 + | \IE_M - \IE ( F ( \underline{X} )) > \lambda\Big)\notag\\
 \leq &~ \IP\, \Big( |F_M (\underline{X}) - \IE_M | > \lambda - 2 \sum_{n=M+1}^{\infty} \sigma_N \Big)\,.
\end{align*}
Since $F_M (\underline{X})$ depends only on finitely many random variables (namely $X_1, \cdots , X_M$), we may apply the known version of the McDiarmid inequality for finite $M$, and obtain 
\begin{align*}
\IP\, \Big( |F_M (\underline{X}) - \IE_M | > \lambda - 2 \sum_{n=M+1}^{\infty} \sigma_N \Big) ~ \leq~ 2 e^{-\frac{2 (\lambda - 2 \sum_{i=M+1}^{\infty} \sigma_i)^2}{\sum_{i=1}^M \sigma_i^2}}
\,.
\end{align*}
Taking the limit $M\to\infty$, we arrive at \eqref{g15}.
\end{proof}

\subsection{Probabilistic estimate}\mbox{}

Now, we estimate the probability that 
$$
F (q_i) := \int_{\Lambda} (\sum_{i \in \IZ^d} q_i f (x - i))^2 ~dx 
$$
 deviates from its mean value.
In the light of \eqref{eq:lade}, this estimates the probability that $\lambda_{1}(H_{\Lambda}^{N})$ is small.

As usual we set  $\Lambda=[-L/2,L/2]^{d}$ and $\Lambda_{1}=[-1/2,1/2]^{d}$.

First, we compute $\IE (F (q_i))$:
\begin{align*}
& \IE \Big( \int_{\Lambda_{1}} \big( \sum_{i\in\IZ^{d}} q_i f (x - i ) \big)^2 ~dx \Big) = \sum_{i,j\in\IZ^{d}} \IE (q_i q_j) \int_{\Lambda_{1}} f (x - i) f (x - j) ~dx\\
& = \sum_{i\in\IZ^{d}} \IV (q_i) \int_{\Lambda_{1}} f (x - i)^2 ~dx + \sum_{i,j\in\IZ^{d}} \IE (q_i) \IE (q_j) \int_{\Lambda_{1}} f (x - i) f (x - j) ~dx\\
& = \IV (q_0) ||f||_2^2 + \IE (q_0)^2 \int_{\Lambda_{1}} \big(\sum f (x - i) \big)^2 ~dx\\
& =: \rho 
\end{align*} 
where $\|f\|_{2}^2:=\int_{\IR^d}|f(x)|^2\;d\,x\,.$ 
Consequently, for integer $L$, we get 
\begin{equation*}
\IE \Big(\int_{\Lambda} V (x) ~dx \Big) = \rho\,{|\Lambda|}\,.
\end{equation*}

To apply Mc Diarmid's inequality, we have to compute the $\sigma_j$. Pick $j \in \IZ^d$, and let 
$Q = \{q_i \}, Q^{'} = \{q_i^{'} \}$ with $q_i = q_i^{'}$, for $i \neq j$, and  $q_j = a ~ q_j^{'} = b$. Then 
\begin{align*}
& |F (Q^{'}) - F (Q) | \leq \int_{\Lambda} \Big| \big( \sum q_i^{'} f (x - i) \big)^2 - \big(\sum q_i f (x - i) \big)^2 \Big| ~dx\\
& \leq \int_{\Lambda} | (b - a) f (x - j)| ~~ | \sum (q_i^{'} + q_i) f (x - i) | ~dx\\
& \leq C \int_{\Lambda} |f (x - j) | ~dx
\end{align*}
where $C = 4 ~ \Big(\sup\; \supp(P_{0})\Big)^2 ~ \underset{x \in \IR}{\sup} ~ \sum |f (x - i) |$.
So, we got to estimate $\int_{\Lambda} | f (x - j) | ~dx$.
\begin{enumerate}
\item[{\em Case} 1:]
$|j| \leq M L$ with $M \geq 3$ to be chosen later. Then we estimate
\begin{equation*}
\int_{\Lambda} |f (x - j) | ~dx \leq \int_{\IR^d} | f (x) | ~dx = ||f||_1 =: \sigma_j \,.
\end{equation*}
\item[{\em Case} 2:]
$|j| > M L$, and hence dist $(j,\Lambda) \geq (M-1) L$. Then
\begin{align*}
\int_{\Lambda} |f (x - j) | ~dx \leq & C \int_{\Lambda} \frac{1}{|x - j|^{\alpha}} ~dx\\
\leq & C L^d \frac{1}{\underset{x \in \Lambda}{\inf} ~ | x - j|^{\alpha}}\\ \leq & C^{'} \frac{L^{d}}{|j|^{\alpha}}\\
=: & \sigma_j\,.
\end{align*}
\end{enumerate}
Therefore, 
\begin{align*}
\sum_{i} \sigma_i^2 ~&\leq~ C_{1} \Big( \sum_{|i| \leq M L} ||f||_1^2 + \sum_{|i| \geq M L} \big(\frac{L^{d}}{|j|^{\alpha}} \big)^2 \Big)\\
&\leq~ C_{2} L^d (1 + L^d L^{- 2 \alpha + d}) \leq C_{3} L^d 
\end{align*}
as $\alpha > d$. Hence, Mc Diarmid's inequality yields 
\begin{align*}
\IP \Big(\Big| \int_{\Lambda} V (x) ~dx - \IE \big(\int_{\Lambda} V (x) ~dx \big) \Big| \geq \lambda L^d\Big) 
\leq & 2 e^{- 2 \frac{\lambda^2 L^{2d}}{\sum \sigma_i^2}}\\ \leq & 2 e^{- 2C_3^{-1} \lambda^2 L^d}\,.
\end{align*}
In particular, whenever $\epsilon < \IE (\int_{\Lambda_{1}} V (x) ~dx)$, and \eqref{eq:EL2} holds true, we have 
\begin{equation} \label{g3} 
\IP \left( \frac{1}{L^d} \int_{\Lambda} V (x) ~dx < \epsilon \right) \leq 2 e^{-\tilde{C} L^d}\leq 2 e^{-\tilde{\tilde{C}} E^{-d/2}}\, .
\end{equation}
Now, combining the upper bound in \eqref{eq:DNbrack}, \eqref{g2}, \eqref{g16}, \eqref{eq:lade}, and \eqref{g3}, 
we find that the upper asymptotic bound  
\begin{equation} \label{g4}
\underset{E\searrow 0}{{\limsup}} ~~ \frac{\ln | \ln N (E)|}{\ln E} ~\leq~- \frac{d}{2}\,
\end{equation} 
holds true under the assumptions of Theorem \ref{th:lifshits} (i). 
\subsection{Upper Bound 2}\mbox{ }

The general upper bound turns out to be correct (i.e. to agree with the lower bound) in the case $\alpha \geq d+1$. For the long range case ($d < \alpha < d+1$) we need another estimate and stronger assumptions. What we need (at least for our proof) is that both $q_i$ and $f$ have a  definite sign. That is why in this section we assume the hypotheses of the second part of Theorem \ref{th:lifshits}. More precisely, 
for definiteness, we assume $\supp ~ P_0 \subset [0, Q]$ and:
\begin{align*}
C_1 (1 + |x|)^{- \alpha} \leq f (x) \leq C_2 (1 + |x|)^{ - \alpha}
\end{align*}
with  $d < \alpha < d+1$ and $C_{1},C_{2}>0$.
We estimate:
\begin{equation} \label{eq:estKS}
\IP \big(\lambda_1 (H_{\Lambda}^N ) < E \big) \leq \IP \big( \underset{x \in \Lambda}{\inf} ~ V_\omega (x) < E \big) 
\leq  \IP \big(\underset{x \in \Lambda}{\inf} U_\omega (x) < E^{\frac{1}{2}} \big)\,. 
\end{equation}
An estimate of the right-hand side of \eqref{eq:estKS} can be found in \cite{KirschS}. For the reader's convenience we give here an alternative proof using Mc Diarmid's inequality.

Setting $\rho=\sqrt{E}$, we estimate
\begin{align*}
& \IP\,\left(\underset{x \in \Lambda}{\inf} \sum_{i \in \IZ^d} q_i f (x - i) < \rho\right) \\
 \leq ~& \IP\, \left( C\; \sum_{i \in \IZ^d} q_i\; \frac{1}{1 + \underset{x \in \Lambda}{\sup} |x - i|^{\alpha}} < \rho\right)\\
 \leq &~ \IP\,\left(C' \sum_{i \in \IZ^d} q_i \frac{1}{(L+|i|)^{\alpha}} < \rho\right)\,.
\end{align*}
We apply McDiarmid's inequality to $\sum q_i \frac{1}{(L + |i|)^{\alpha}}$, then with

\begin{align}
\sigma_j = 2 Q \frac{1}{(L + |j|)^{\alpha}} \leq
\begin{cases}
\frac{C}{L^{\alpha}} & \text{~for~} |j| \leq M L\,,\\
\frac{C}{|j| \alpha} & \text{~ for ~} |j| > M L\,.
\end{cases}
\end{align}

Thus, 
\begin{align}
   \sum \sigma_j^2 \leq C^{'} L^{d - 2 \alpha}\,.
\end{align}

Moreover,
\begin{equation*}
\IE \big( \sum q_i \frac{1}{(L + |i|)^{\alpha}} \big) = C \sum \frac{1}{(L + |i|)^{\alpha}} \sim L^{d - \alpha}\,.
\end{equation*}
So, we have to take $\rho \sim L^{d - \alpha} $.
With this choice, McDiarmid's inequality gives

\begin{align}
\IP \left(\sum q_i \frac{1}{(L + |i|)^{\alpha}} < C L^{d - \alpha} \right) ~\leq &~ e^{- C^{'} L^d}\notag\\
\leq~& e^{- C^{''} \rho^{- \frac{d}{\alpha - d}}}\notag\\ ~= &~ e^{-C^{''}E^{- \frac{d}{2 (\alpha - d)}}}\,.\label{g17} 
\end{align}
This estimate is better than the general estimate \eqref{g3} 
if
$2 ( \alpha - d) < 2$, that is, if $\alpha < d+1$.

Putting together the upper bound in \eqref{eq:DNbrack}, \eqref{g2}, \eqref{g16}, \eqref{eq:estKS}, and \eqref{g17},
we conclude that under the assumptions of Theorem \ref{th:lifshits} (ii), we have 
\begin{equation} \label{g5}
\underset{E\searrow 0}{{\limsup}} ~~ \frac{\ln | \ln N (E)|}{\ln E} ~\leq~- \frac{d}{2(\alpha-d)}\,.
\end{equation}

\section{Lower bound}\label{sec:LB}
In this section we suppose that Assumption \ref{Ass:2} holds true, i.e that 
    \begin{equation} \label{g20} 
\IP (|q_i| \leq \epsilon ) \geq C \epsilon^K 
    \end{equation} 
 for some $C > 0$, $K \geq 0$, and all $\epsilon > 0$ small enough. Without loss of generality we assume that $C=1$ and $K>0$.

Now, we consider $H_{\Lambda}^D$ with Dirichlet boundary conditions. By the lower bound in \eqref{eq:DNbrack}, we have 
\begin{equation} \label{g6} 
N(E) \geq L^{-d} \IP\left(\lambda_1(H_\Lambda^D) \leq E\right). 
\end{equation}

For further references we recall that the ground state energy $\lambda_0$ of the DIrichlet Laplacian $- \Delta_{\Lambda}^D$ is given by 
\begin{equation} \label{g1} 
\lambda_0 = d \left(\frac{\pi}{L}\right)^2
\end{equation}
and the ground state is
\begin{equation*}
\varphi_0 (x) = \left(\frac{2}{L}\right)^{\frac{d}{2}} \prod_{i=1}^d \cos (\pi \, L^{-1} x_i), \quad x \in \Lambda \,.
\end{equation*}
We consider the set $\Omega_L^{\epsilon} \subset \Omega$ with
\begin{equation*}
\Omega_L^{\epsilon} = \{ \epsilon ~ | ~~ |q_i| \leq \epsilon \text{~for~} |i| \leq L + R \}, \quad R \geq L\,.
\end{equation*}
Later we will choose $R$  as $R = L^{\beta}$ with $\beta \geq 1$. By \eqref{g20}, 
\begin{equation*}
\IP (\Omega_L^{\epsilon}) \geq (\epsilon^K)^{(L + R)^d} = e^{K (\ln ~ \epsilon) (L + R)^d}\,.
\end{equation*}
We will show that $\lambda_1 (H_{\Lambda}^D)$ is small on $\Omega_L^{\epsilon}$. We have
\begin{align*}
\lambda_1 (H_{\Lambda}^D)&  \leq ~~ < \varphi_0, H_{\Lambda}^D \varphi_0 >\\ ~& = \lambda_0 + \int_{\Lambda} V \varphi_0^2 ~dx\\
& \leq~ \lambda_0 + \frac{C}{|\Lambda^{'}|} \int_{\Lambda} V (x) ~dx, 
\end{align*}
$\lambda_0  \approx L^{-2}$ being given by \eqref{g1}.

Now for $\omega \in \Omega_L^{\epsilon}$
\begin{align*}
& \frac{1}{|\Lambda|} \int_{\Lambda} V (x) ~dx\\
 \leq ~& \frac{1}{|\Lambda|} \int_{\Lambda} \Big( \epsilon \sum_{|i| \leq L + R} |f (x - i)| + Q \sum_{|i| > L + R} |x - i|^{-\alpha} \Big)^2 ~dx\\
\leq ~& \frac{2 \epsilon}{|\Lambda|} \int_{\Lambda} \Big( \sum_{|i| \leq L + R} |f (x - i)| \Big)^2 ~dx + \frac{2 Q}{|\Lambda|} \int_{\Lambda} \Big( \sum_{|i| > L + R} |x - i|^{-\alpha} \Big)^2 ~dx\\
 \leq ~& C \Big( \epsilon + \underset{x \in \Lambda}{\sup} ~ \big( \sum_{|i| > L + R} |x - i|^{-\alpha} \big)^2 \Big)\\
 \leq ~& C^{'} \Big( \epsilon + \big(\sum_{|i| > L+R} |i|^{-\alpha}\big)^2 \Big)\\
 \leq ~& C^{''} \left( \epsilon + R^{- 2 \alpha + 2 d} \right)\,.
\end{align*}
Let us now choose $R = L^{\beta}$. If $\alpha \geq d + 1$, we take $\beta = 1$. Then
\begin{equation*}
\frac{1}{|\Lambda|} \int_{\Lambda} V ~dx \leq C^{''} \big( \epsilon + L^{- 2 \alpha + 2 d} \big) \leq C^{''} \epsilon + L^{-2}. 
\end{equation*}
Thus, for $\omega \in \Omega_L^{L^{-2}}$, we have
\begin{equation} \label{g18} 
\lambda_1 (H_{\Lambda}^D) \leq C L^{-2}\,.
\end{equation}
Hence, 
\begin{equation} \label{g19} 
\IP \left(\lambda_1 (H_{\Lambda}^D) \leq C L^{-2} \right)  \geq \IP (\Omega_L^{L^{-2}})
 \geq e^{C (\ln L^{-2}) L^d}\, .
\end{equation}
Choosing $L$ so that $E = CL^{-2}$, we find that  
\begin{equation} \label{g7} 
\IP \left(\lambda_1 (H_{\Lambda}^D) \leq E \right)  \geq 
C' e^{C'' (\ln E) E^{-d/2}}. 
\end{equation} 
Putting together \eqref{g6} and \eqref{g7}, we get 
$$ 
\underset{E\searrow 0}{{\liminf}} ~~ \frac{\ln | \ln N (E)|}{\ln E} ~\geq~- \frac{d}{2}\,,
$$ 
which combined with \eqref{g4} implies \eqref{g8}. This completes the proof of the first part of Theorem \ref{th:lifshits}.

Now, we turn to the case $d < \alpha < d+1$. In this case we take $\beta = \frac{1}{\alpha - d} > 1$. Then, similarly to \eqref{g18}, we have 
\begin{equation*}
\lambda_1 (H_{\Lambda}^D) \leq \tilde{C} L^{-2} 
\end{equation*}
for $\omega \in \Omega_L^{L^{-2}}$. Therefore, similarly to \eqref{g19} we get 
\begin{equation*}
\IP (\lambda_1 (H_{\Lambda}^D) \leq \tilde{C} L^{-2} ) \geq P (\Omega_L^{L^{-2}}) 
\geq e^{C (\ln L^{-2}) R^d} \geq e^{ C^{'} (\ln L^{-2}) L ^{\frac{d}{\alpha - d}}}\,.
\end{equation*}
Choosing $L$ so that  $E = C' L^{-2}$, we obtain 
\begin{equation*}
\IP (\lambda_1 (H_{\Lambda}^D) \leq E) \geq \tilde{C}e^{\tilde{C}^{''} (\ln E)\, E^{\frac{d}{2(\alpha - d)}}} \,,
\end{equation*}
which, together with \eqref{g6} yields 
    \begin{equation} \label{g10} 
\underset{E\searrow 0}{{\liminf}} ~~ \frac{\ln | \ln N (E)|}{\ln E} ~\geq~- \frac{d}{2(\alpha - d)}\,. 
    \end{equation} 
Now, \eqref{g9} follows from \eqref{g5} and \eqref{g10}, and the proof of Theorem \ref{th:lifshits} (ii) is complete.\\

{\bf Acknowledgements}. The authors gratefully acknowledge  the partial
support of the Chilean Scientific Foundation {\em Fondecyt} under Grants 1130591  and 1170816.\\
A considerable part of this work has been done during W. Kirsch's visits to the Pontificia Universidad Cat\'olica de Chile and during  G. Raikov's visits to the University of Hagen,
Germany. We thank these institutions for financial support and hospitality.

We also would like to thank the referee for careful reading and a number of useful suggestions which, we feel, improved substantially the readability of the article.
\\
\bibliographystyle{plain}
\bibliography{LifshitzC}

\vspace{1cm}
{\sc Werner Kirsch}\\
Fakult\"at f\"ur Mathematik und Informatik\\
FernUniversit\"at in Hagen\\
Universit\"atsstrasse 1\\
D-58097 Hagen, Germany\\
E-mail: werner.kirsch@fernuni-hagen.de\\

{\sc Georgi Raikov}\\
Facultad de Matem\'aticas\\
Pontificia Universidad Cat\'olica de Chile\\
Av. Vicu\~na Mackenna 4860\\ Santiago de Chile\\
E-mail: graikov@mat.uc.cl

\end{document}